\documentclass[11pt]{article}

\usepackage{amsthm,amssymb,amsfonts,amsmath}
\usepackage{amscd} %Commutative diagrams. 
\usepackage{mathrsfs}
\usepackage{graphicx}
\usepackage{color}
\usepackage{hyperref}
\usepackage{bm}
\usepackage[lined, ruled, linesnumbered]{algorithm2e}
\usepackage[margin=1.25in]{geometry}
\graphicspath{ {./Images/} }
\usepackage{authblk}

\usepackage[backend=bibtex,style=ieee-alphabetic,natbib=true]{biblatex} %added
\newcommand{\I}[1]{{\mathbf 1}_{\left[#1\right]}}

\newtheorem{thm}{Theorem}
\newtheorem{lem}[thm]{Lemma}

\newtheorem{cor}[thm]{Corollary}
\newtheorem{dfn}[thm]{Definition}

\numberwithin{thm}{section}
\numberwithin{equation}{section}
\begin{document}
\title{Perfect Sampling in Infinite Spin Systems\\ via Strong Spatial Mixing}
%MRJ  Added acknowledgement.
\author{Konrad Anand\thanks{Supported by a studentship from the School of Mathematical Sciences.} }
\author{Mark Jerrum\thanks{Supported by grant EP/S016694/1 `Sampling in hereditary classes' from the Engineering and Physical Sciences Research Council (EPSRC) of the UK.}}
\affil{School of Mathematical Sciences\\Queen Mary, University of London\\Mile End Road, London E1 4NS\\United Kingdom}
\date{}
\maketitle

\begin{abstract}
We present a simple algorithm that perfectly samples configurations from the unique Gibbs measure of a spin system on a potentially infinite graph~$G$.  The sampling algorithm assumes strong spatial mixing together with subexponential growth of~$G$.  It produces a finite window onto a perfect sample from the Gibbs distribution. The run-time is linear in the size of the window. 
\end{abstract}

%\newpage
%\tableofcontents
%\newpage

\section{Introduction}

Many interesting computational problems have come out of statistical physics, in particular from the  study of spin systems. A wide variety of systems used and studied in probability, physics, machine learning, and theoretical computer science fall under this domain.  In brief, a spin system is defined by a graph~$G$, a finite (in our case) set of spins, and a local `specification'.  A configuration of the system is an assignment of spins to the vertices of the graph.  The local specification defines a probability distribution on the configurations of the system.  Although the specification is local, this Gibbs distribution can exhibit long-range effects.  A primary algorithmic goal in this area is to sample configurations from the Gibbs distribution.

Spin systems have been an active area of research from a computational perspective in the past few decades, with many algorithms proposed and many results on hardness achieved. Many of the most successful algorithms have been based on Markov chain simulation.  Particularly prominent is the familiar Glauber dynamics \cite{glauber1963} whose mixing time has remained an active area of research, with wonderful new results coming in recent years \cite{Chen2020OptimalMO, chen2020rapid}, prompted by pioneering work on high dimensional expanders \cite{KaufmanOppenheim,ALOV}. 
%Other Markov chain based algorithms have solved problems such as efficiently sampling from the monomer-dimer model \cite{Jerrum1989ApproximatingTP} and efficiently sampling from the ferromagnetic Ising model \cite{Jerrum1993PolynomialTimeAA}, both being examples of spin systems themselves.

Deterministic approaches have also been proposed, based on decay of correlations \cite{weitz2006counting} or Taylor expansion in a zero-free region of the partition function \cite{PatelRegts,BarvinokSoberon}.  As with Markov chain simulation, these approaches provide a sample from an approximation to the desired distribution.  Another strand of research has been into perfect samplers.  Aside from the direct algorithmic challenge of producing perfect samples, there is another, more practical, motivation.  This motivation arises not from the exactness of the sampling distribution as such;  after all, the $t$-step distribution of a Markov chain converges exponentially in~$t$ to stationarity, so the deviation of the sampling distribution from the desired one can been made very small at modest computational expense.  More significant is having a definite termination condition, in contrast to Markov chain simulation which has to be carried on for `sufficiently many' steps.  This requires an a priori bound on the `mixing time' of a Markov chain, which must be analytically derived and may be much larger that actually necessary.  In contrast, perfect samplers can respond to characteristics of a problem  instance, and produce results more quickly on favourable instances.

Initially, it is not obvious that perfect sampling is possible in complex settings.  The first demonstration came with the Coupling From The Past (CFTP) approach of Propp and Wilson \cite{ProppWilson}.  There have been other recent approaches to the problem of efficient perfect sampling which move away from Markov chains, such as the partial rejection sampling of Guo, Jerrum and Liu \cite{guo2019uniform} and the randomness recycler of Fill and Huber \cite{fill2000randomness}.  Our contribution is in a similar direction.

Our particular goal in this work is to produce perfect samples from a spin system, even when  the underlying graph~$G$ may be infinite.  Although the specification of a spin system is local, it may, if the conditions are right, define a unique Gibbs measure on an infinite graph~$G$.  Classic examples of this phenomenon include the Ising and hardcore models on $\mathbb{Z}^d$, in the subcritical regime.  Since we cannot write down an infinite spin configurations, we need to explain what we mean by perfect sampling in this situation.  Consider a spin system on $\mathbb{Z}^2$ that has a unique Gibbs measure.  As an example of what we want to achieve, we might cite the problem of sampling (perfectly) from the the distribution of spins on an $L\times L$ square fragment of $\mathbb{Z}^2$ induced by the infinite Gibbs measure.   

Van den Berg and Steif \cite{vdBerg1999codings} pioneered perfect sampling from infinite Gibbs measures.  They showed that perfect sampling is possible in the case of the ferromagnetic Ising model on $\mathbb{Z}^2$ in the uniqueness regime, i.e., when there is a unique Gibbs measure.  They also gave a precise negative result that essentially rules out perfect simulation in non-uniqueness.  Their approach is through CFTP{}.  Spinka~\cite{spinka2020finitary}, in a wide-ranging investigation, has taken forward this approach.  In contrast we depart from CFTP, and take our inspiration from partial rejection sampling and related work.

The closest prior work to ours is that of Feng, Guo and Yin \cite{Feng2019PerfectSF}.  Like them, our key assumptions are that the spin system displays `strong spatial mixing', and that the underlying graph has `subexponential growth'.  Definitions of these terms can be found in Section~\ref{sec:requirements}.  Roughly speaking, we require correlations between spins to decay rapidly with distance, and local neighbourhoods of~$G$ to grow slowly enough with increasing radius.  Taken together, these conditions allow us to largely discount the effect of spins at far away vertices, which is key to the efficiency of our algorithm.  The conditions of strong spatial mixing and subexponential growth are quite standard and have been widely studied \cite{Goldberg2005StrongSM, weitz2006counting}.  Various authors have sought to leverage the property of strong spatial mixing to algorithmic ends \cite{Feng2019PerfectSF, spinka2020finitary, sinclair2013spatial,sinclair2017spatial}.  Intriguingly, it has been suggested \cite{Chen2020OptimalMO} that strong spatial mixing may be equivalent to spectral independence, a concept that is currently the subject of intensive study.  As for the second condition, subexponential growth currently seems important for perfect sampling, even though some approaches to approximate sampling dispense with it.

The main difference between our approach and that of Feng et al.\ is the following.  Their strategy is to start with a spin configuration on the whole of $G$, and gradually `repair' the configuration until it becomes a perfect sample.  As this approach seems inapplicable to infinite graphs, we instead `grow' the perfect sample one vertex at a time.  Our recursive algorithm tries to guess the spin at a vertex~$v$ without looking at spins that live on a sphere of radius~$\ell$ centred at~$v$; when it cannot guess, the algorithm recursively determines the spins on the sphere before returning to assign a spin to~$v$.

As it runs, algorithm traces out something akin to a branching process, whose expected size is finite with probability~1 when sufficiently strong spatial mixing holds. When successful, the algorithm is linear time:  for each vertex whose spin we sample, the algorithm takes constant time.  An advantage to the algorithm is that we do not need a configuration on an entire graph to determine the spin at a single vertex. This gives the possibility of drastically more efficient sampling than any algorithm requiring an entire configuration, and allows us to sample from infinite graphs as well. The algorithm provides a novel perspective on the use of strong spatial mixing on spin systems.  Indeed, the algorithm could be viewed as an effective proof that strong spatial mixing together with subexponential growth implies uniqueness of the Gibbs measure.

\section{Spin systems}
For the time being, we restrict attention to finite graphs, and defer the discussion of infinite graphs until later.
%MRJ  Simplified the next 2 sentences.
Given a finite graph $G = (V,E)$ and $q \in \mathbb N$ we work with the set of configurations
$\Omega_V = [q]^V$.
When there is no room for confusion, we will write $\Omega$ instead of~$\Omega_V$.  Naturally, when $W\subset V$ we write $\Omega_W=[q]^W$ for the partial configurations restricted to~$W$.

For $\sigma \in \Omega$ and $W \subset V$ we denote its restriction to $W$ by $\sigma_W$. For simplicity, when $W= \{v\}$ we will write $\sigma_v$ instead of $\sigma_{\{v\}}$. In this same situation, we denote the spin at $v$ by $i$, and write $\sigma_v=i$ rather than $\sigma_v=(i)$.

Given a field $b : [q] \to \mathbb R$ and a symmetric interaction weight $A: [q] \times [q] \to \mathbb R$, the Gibbs distribution gives to each configuration~$\sigma$ the weight
\begin{align*}
  \prod_{v \in V} b (\sigma_v) \prod_{(u,v) \in E} A(\sigma_u,\sigma_v).
\end{align*}
Define the partition function $Z(G)$ to be the sum of all possible weights
\begin{align*}
  Z(G) := \sum_{\sigma \in \Omega} \prod_{v \in V} b (\sigma_v) \prod_{(u,v) \in E} A(\sigma_u,\sigma_v).
\end{align*}
Then we define the Gibbs measure to be the measure $\mu : \Omega \to [0,1]$ where
\begin{align}\label{eq:specification}
  \mu(\sigma) = \frac{\prod_{v \in V} b (\sigma_v) \prod_{(u,v) \in E} A(\sigma_u,\sigma_v)}{Z(G)}.
\end{align}
Frequently in this paper we will want to look at the marginal distribution induced by fixing the configuration on some vertices. For any $W$, define the marginal distribution of $\mu$ on $W$ by 
$$
\mu_W(\sigma)=\sum_{\sigma'\in\Omega:\sigma'_W=\sigma}\mu(\sigma'),
$$
for all $\sigma \in \Omega_W$.  Also, for any $\Lambda \subset V$ and $\tau \in \Omega_\Lambda$, define the marginal distribution with boundary condition $(\Lambda,\tau)$ by 
\begin{align}\label{eq:conditional}
  \mu_W^{(\Lambda,\tau)}(\sigma) :=
     \frac1{\mu_{\Lambda}(\tau)} \sum_{\substack{\sigma'\in\Omega:\\ \sigma'_\Lambda = \tau, \sigma'_W = \sigma}}\!\!\mu(\sigma'),
\end{align}
assuming $\mu_\Lambda(\tau)\not=0$.
We abbreviate $\mu_W^{(\Lambda,\tau)}$ to $\mu_W^\tau$ when the set $\Lambda$ is clear from the context.
With this definition of the Gibbs measure, we have a fixed probability space determined by a few parameters. We call this a spin system $\mathcal S = (G, q, b, A).$

As usual, the \emph{Gibbs property} of $\mu$ is important to us.  Suppose $W\subset V$ and that the configuration $\tau\in\Omega_{V\setminus W}$ is \emph{feasible}, i.e., $\mu_{V\setminus W}(\tau)>0$.  Define the \emph{boundary} of~$W$ by 
    $$
    \partial W=\{v\in V\setminus W:\{u,v\}\in E\text{ for some }u\in W\}, 
    $$
and let $\tau'\in\Omega_{\partial W}$ be the restriction of $\tau$ to $\partial W$.  Then $\mu_W^{(V\setminus W,\tau)}=\mu_W^{(\partial W,\tau')}$; in other words, conditioning on the complement of $W$ is equivalent to conditioning on its boundary.

\section{Sketch of the basic algorithm}

%MRJ  Various small changes to make clear than y is selected once and then reused.

Before describing the full algorithm and the properties it requires of the spin system, we pause to describe a basic version of the algorithm, the `radius 1' version of the full algorithm.  After this initial consideration, we will expand the algorithm to a radius determined by spatial mixing.

Suppose we start with a graph $G = (V,E)$ and we want to know the spin at a single vertex $v \in V$. We give this vertex a realisation $y$ of a $U[0,1]$ (i.e., uniform in the unit interval) random variable. If we knew the spins of the neighbours of $v$ it would be a simple task to partition the unit interval according to the probabilities of each possible spin at~$v$, and then use $y$ to determine the spin at~$v$.  Unfortunately, we do not and so we need some way around this.

Instead of just partitioning $[0,1]$ according to the respective spin probabilities, we partition it according to the minimum probability for each spin given any configuration on $v$'s neighbours. Let $W = \{w_1,...,w_d\}$ be the neighbours of $v$. For $i \in [q]$ define
\begin{align*}
  p_v^i = \min_{\tau \in \Omega_W} \mu_v^\tau(i).
\end{align*}
So $p_v^i$ is the minimum, over all assignments of spins to neighbours of~$v$, of the probability that vertex~$v$ should be assigned spin~$i$. 
For the purposes of the algorithm we will partition $[0,1]$ into intervals $I_1,...,I_q$ of sizes $p_v^1,...,p_v^q$ with a remaining interval which we call the `zone of indecision' $I_0$ of size $p_v^0 = 1 - \sum_{i = 1}^q p_v^i$.

If the point $y\in[0,1]$ selected earlier falls in $\cup_{i = 1}^q I_i$ then we can assign a spin to $v$. If not, we need to know the spins of the neighbours in order to split $I_0$ into different spin regions. We will call the above process on each of the neighbours recursively in turn.

How we make these recursive calls is important---the outcome for $W$ must follow the marginal distribution $\mu_W$, not the product of the marginal distributions at each neighbour, that is, for $\sigma \in \Omega_W$ we want
\begin{align*}
  \mu_W(\sigma) = \mu_{w_1}(\sigma_{w_1}) \cdot \mu_{w_2}^{\sigma_{w_1}}(\sigma_{w_2}) \cdot \cdot \cdot \mu_{w_d}^{\sigma_{W - w_d}} (\sigma_{w_d}).
\end{align*}

To this end, when we call the algorithm on $w_j$, we condition on the spins of $w_1,...,w_{j-1}$ to ensure the correctness of the probabilities. Once we have the spins at each of the neighbours, we subdivide the zone of indecision $I_0$ into $q$~intervals corresponding to the different spins.  We now use $y$ to determine the spin at $v$ and discard the intermediate working (including any spins applied to vertices other than~$v$) that we have used during the algorithm.

Why should this strategy work?  Well, if we chose the spins $\sigma_W$ first, then subdivided the interval $[0,1]$ appropriately, and then chose a uniform random point~$y$ in the interval, we would certainly assign a spin to~$v$ with the correct probability.  We merely observe that, if we choose the random point $y$ first, then at least some of the time we do not need to know $\sigma_W$ and we can omit the first step.  If the probability of omitting the first step is large enough, we can avoid an infinite recursion.

\subsection{Example: the hardcore model}\label{sec:hardcore}

A classic example of a spin system is the hardcore model. The hardcore model is a distribution over independent sets of a graph, which we model as a spin system with $q=2$ spins.  For consistency with the notation already introduced, we label these spins $1$ and $2$ rather than the usual $0$ and $1$.  The one constraint is that neighbouring vertices are not allowed to both have spin~2. The hardcore distribution gives independent sets a weight depending on their size. To be explicit, the associated field and interaction weights are:

\begin{align*}
  b = 
\begin{pmatrix}
  1 \\
  \lambda
\end{pmatrix}, \quad A =
\begin{pmatrix}
  1 && 1 \\
  1 && 0
\end{pmatrix}.
\end{align*}

Now, to see the basic algorithm in action, we consider its behaviour on the hardcore model. Calling the algorithm at a vertex~$v$, we see that the $p_v^1 = \frac1{1+\lambda}$ and $p_v^2 = 0$.  (If one of the neighbours of~$v$ has spin~2, then we cannot assign spin~2 to vertex~$v$.  So $p_v^2 = 0$.  On the other hand, if all the neighbours have spin~1 there is just a probability $\frac1{1+\lambda}$ that we will assign spin~1 to $v$.  So $p_v^1 = \frac1{1+\lambda}$.)  Thus we set the spin to~1 immediately with probability $\frac1{1+\lambda}$ and call the algorithm recursively on the neighbours with probability $\frac\lambda{1 + \lambda}$. 

Consider the recursive calls on the neighbours, $w_1,...,w_d$. If the algorithm sets the spin of any $w_i$ to 2 then we set the spin of $v$ to~$1$ and terminate the algorithm. Otherwise, the neighbouring vertices all have spin 1. Conditioned on this event, $v$ has probability $\frac1{1 + \lambda}$ of being 1 and probability $\frac\lambda{1 + \lambda}$ of being 2. Thus the entire zone of indecision is moved to the case of the spin being $2$ and we terminate with this assignment.

We can view the recursive calls of the algorithm as a random tree $T_A$. By comparing $T_A$ to a branching process, we will see that for graphs with maximum degree $\Delta$, when $\lambda < \frac1{\Delta - 1}$ the expected time to set the spin of one vertex is $O(1)$.

Each vertex in $T_A$ will have no children when $y \in I_1 \cup I_2$ and will have a child for each neighbour when $y \in I_0$. The number of neighbours is bounded by $\Delta$ and the probability of having children is $p_v^0 = \frac\lambda{1+\lambda}$, so we can bound the size of $T_A$ by a branching process with the offspring distribution $\xi$ where 
\begin{align*}
    \mathbb P[\xi = 0] = \frac1{1 + \lambda}, \quad \mathbb P[\xi = \Delta] = \frac\lambda{1 + \lambda}.
\end{align*}
The expected size of such a branching process is finite when
\begin{align*}
  \Delta \cdot\frac\lambda{1+\lambda} < 1 \iff \lambda < \frac1{\Delta - 1}.
\end{align*}
To be exact, the expected size of this branching process---and the expected run-time of the algorithm---is
%MRJ  Slightly simplified the fraction.
\begin{align*}
  \frac1{1-\Delta \cdot\frac\lambda{1+\lambda}} = \frac{1+ \lambda}{1 -(\Delta-1) \lambda}.
\end{align*}
Thus, the size of $T_A$ is finite when $\lambda < \frac1{\Delta - 1}$ and in this case the expected run-time of the algorithm is bounded by $\frac{1+ \lambda}{1 -(\Delta-1) \lambda}$.

\section{Spin systems on infinite graphs:  our requirements}\label{sec:requirements}
Up to this point, our graphs have been finite, but we now want to extend the discussion to infinite graphs.  Although we now allow the vertex set to be countably infinite, we do assume that graphs are locally finite, i.e., that all vertex degrees are finite.  For perfect sampling to make sense, we need the system we are dealing with to have unique infinite Gibbs measure.  
%In our algorithm we will call on the Gibbs measure of our spin system. 
%This is uncomplicated provided that the underlying graph is finite, but not so when the graph is infinite. 
Briefly, a Gibbs measure $\mu$ for an infinite graph $G$, with interaction weights~$A$ and field~$b$, is a probability measure that has the following properties:
\begin{itemize}
    \item Events that are determined by a finite set of spin variables have a well defined probability of occurrence.  Technically, cylinder events are measurable.
    \item For any finite vertex subset $W\subset V$, the probability distribution of $\sigma_W$ (i.e., the marginal distribution on~$W$) is consistent with that given by equation \eqref{eq:specification} for the finite graph $G[W]$.  Specifically, suppose $W\subset V$ is any finite vertex subset, and define the boundary of~$W$ by
    $$
    \partial W=\{v\in V\setminus W:\{u,v\}\in E\text{ for some }u\in W\}. 
    $$
    Then, for all $\tau\in\Omega_{\partial W}$, the probability distribution $\mu_W^{(\partial W,\tau)}$ (which is well defined as it is determined by a finite set of spin variables) agrees with equation~\eqref{eq:conditional}, where $\Omega=[q]^{W\cup\partial W}$ and $\Lambda=\partial W$.
\end{itemize}
%We can only call on the Gibbs measure in an unambiguous way when we are in a uniqueness regime, that is there is a unique Gibbs measure compatible with the Gibbs measure over all finite subgraphs of the entire graph (see \citep{friedli_velenik_2017}. 
When there is a unique infinite Gibbs measure satisfying the above conditions we say that the system is in the `uniqueness regime'. For a thorough treatment of these ideas see, e.g., Friedli and Velenik~\citep{friedli_velenik_2017}.

Luckily, we will require strong spatial mixing for the algorithm, which itself implies that there is a unique infinite measure~$\mu$ with the above properties.  A perfect sampling algorithm for~$\mu$ is required to open up a finite window onto a perfect sample from~$\mu$. 

The basic algorithm, which was sketched in the context of the hardcore model, made recursive calls on the immediate neighbours of vertex~$v$.  It is correct, but works for only small values of $\lambda$.  We therefore analyse a more general, radius~$\ell$, version of the algorithm, which recurses on vertices at distance $\ell$ from $v$, rather than its immediate neighbours.  Let $d_G$ denote graph distance and 
\begin{align*}
  S_\ell=S_\ell(v) := \{w \in V  : d_G(v,w) = \ell\}
\end{align*}
denote the sphere of radius~$\ell$ about vertex~$v$. Let $(\Lambda,\sigma)$ be the partial assignment of existing spins.  Then we redefine probabilities with which we partition the unit interval. Let
\begin{align*}
    p_v^i &:= \min_{\tau \in \Omega_{S_\ell\setminus\Lambda}} \mu_v^{(\Lambda,\sigma) \oplus(S_\ell\setminus\Lambda,\tau)}(i) \quad \forall i = 1,...,q, \\
  p_v^0 &:= 1 - \sum_{i \in [q]} p_v^i.
\end{align*}
Here, $\oplus$ denotes `concatenation' of spins.  Thus $(\Lambda,\sigma) \oplus(S_\ell\setminus\Lambda,\tau)$ is an assignment of spins to $\Lambda\cup S_\ell$ that agrees with $\sigma$ on $\Lambda$ and $\tau$ on $S_\ell\setminus\Lambda$. We will simplify this notation to $\sigma\oplus\tau$ if the domains of $\sigma$ and~$\tau$ are clear from the context.

To analyse the efficiency of the algorithm, we consider its behaviour as a branching process where at a given vertex $v$ the probability of no children is $1 - p_v^0$ and the probability of deg$(v)$ children is $p_v^0$. We want to limit this branching, and to do so we leverage strong spatial mixing. Say that a partial assignment $(W,\tau)$ is \emph{feasible} if $\mu_W(\tau)>0$.  Also denote total variation distance between two distributions by $d_\mathrm{TV}$. 

\begin{dfn}
Let $\mathcal S$ be a spin system and $f : \mathbb N \to \mathbb R^+$ a function.  Given $W\subset V$ and $\tau^1,\tau^2\in\Omega_W$, let $\ell = \min\{d_G(v,w) : w \in W, \tau_w^1 \not= \tau_w^2\}$. If for all $W$ with $V\setminus W$ finite, all $v\in V\setminus W$, and all feasible $\tau^1, \tau^2 \in \Omega_W$ with $\tau^1\not=\tau^2$,
\begin{align*}
  d_\mathrm{TV}\left(\mu_v^{\tau^1},\mu_v^{\tau^2} \right) \leq f(\ell),
\end{align*}
then we say that $\mathcal S$ exhibits \emph{strong spatial mixing} with rate~$f$.
\end{dfn}
Note that there is no problem giving a meaning to the conditional distributions $\mu_v^{\tau^1}$ and $\mu_v^{\tau^2}$:
since $V\setminus W$ finite, the Gibbs property of $\mu$ ensures that $\mu_v^{\tau^1}$ and $\mu_v^{\tau^2}$ depend on a finite set of spins.
We want to use this bound to control the size of the zone of indecision which branches the algorithm.

\begin{lem}\label{lem:pvzerobound}
    For any $\Lambda\subset V$, $v \in V\setminus\Lambda$, and feasible partial configuration $(\Lambda,\sigma)$, 
	\begin{align*}
	  p_v^0 \leq q \cdot \max_{\tau^1, \tau^2 \in \Omega_{S_\ell(v)\setminus\Lambda}} d_{\mathrm{TV}}\left(\mu_v^{\sigma \oplus \tau^1}, \mu_v^{\sigma \oplus \tau^2} \right),
	\end{align*}
	where the maximisation is over all feasible $\sigma\oplus\tau^1$ and $\sigma\oplus\tau^2$. 
\end{lem}

\begin{proof}
For convenience, set $W=S_\ell(v)\setminus\Lambda$.
Let $\tau' \in \Omega_W$ be an arbitrary configuration with $\sigma\oplus\tau'$ feasible. Then
\begin{align*}
  p_v^0 = 1 - \sum_{i \in [q]} p_v^i = \sum_{i \in [q]} \left(\mu_v^{\sigma\oplus\tau'}(i) - p_v^i \right) \leq \sum_{i \in [q]} \max_{\tau \in \Omega_W} \left(\mu_v^{\sigma\oplus\tau}(i) - p_v^i\right).
\end{align*}
Clearly, for every $i \in [q]$ it holds that
\begin{align*}
  \max_{\tau \in \Omega_W} \left(\mu_v^{\sigma\oplus\tau}(i) - p_v^i\right) \leq \max_{\tau^1, \tau^2 \in \Omega_{W}} d_{\mathrm{TV}} \left(\mu_v^{\sigma\oplus\tau^1},\mu_v^{\sigma\oplus\tau^2} \right),
\end{align*}
so it follows that
\begin{align*}
  p_v^0 \leq \sum_{i \in [q]} \max_{\tau \in \Omega_W} \left(\mu_v^{\sigma\oplus\tau}(i) - p_v^i \right) \leq q \cdot \max_{\tau^1, \tau^2 \in \Omega_W} d_{\mathrm{TV}}\left(\mu_v^{\sigma\oplus\tau^1}, \mu_v^{\sigma\oplus\tau^2} \right).
\end{align*}
\end{proof}

It is worth pointing out that we cannot do much better than this bound. Consider the family of distributions $\{\mu_j : j \in [q]\}$ over $[q]$ such that
\begin{align*}
  \mu_j(i) = 
  \begin{cases}
  	0 &: \text{if}\ i = j, \\
  	\frac 1{q-1} &: \text{otherwise}.
  \end{cases}
\end{align*}
Then $p_v^0 = 1$ and for all $j,k \in [q], j \not= k$
\begin{align*}
  d_{\mathrm{TV}} (\mu_j, \mu_k) = \frac1{q-1}.
\end{align*}
We will leverage this bound and strong spatial mixing for the full algorithm.

Finally, a useful property of a graph is \textit{sub-exponential neighbourhood growth}.
\begin{dfn}
Let $G = (V,E)$. We say that $G$ has sub-exponential neighbourhood growth if for all $c > 1$ there exists an $N \in \mathbb N$ such that for all $v \in V$ and all $\ell > N$
\begin{align*}
  |S_\ell(v)| < c^\ell.
\end{align*}
\end{dfn}
This property provides a restriction on how fast the number of recursive calls in our algorithm grows.

\section{The algorithm in detail}

The general algorithm is quite similar to the basic algorithm, but rather than conditioning on neighbours we condition on $S_\ell$, the sphere of radius~$\ell$ around our chosen vertex.

The algorithm takes a spin system on a graph~$G$ with a partial configuration $\sigma$ on some subset~$\Lambda$ of the vertices. (Informally, the partial configuration $(\Lambda,\sigma)$ represents the spins that have already been determined.)  The algorithm is then called on some unassigned vertex~$v$ with a distance parameter $\ell$.  Recall that $S_\ell(v) := \{w \in V: d_G(v,w) = \ell\}$ is the sphere of radius~$\ell$ centred at vertex~$v$.  Recall also the probabilities with which we split the unit interval used to decide the spin at~$v$:
\begin{align*}  \label{ps}
  p_v^i &:= \min_{\tau \in \Omega_{S_\ell\setminus\Lambda}} \mu_v^{(\Lambda,\sigma)\oplus(S_\ell\setminus\Lambda,\tau)}(i) \quad \forall i = 1,...,q, \\
  p_v^0 &:= 1 - \sum_{i \in [q]} p_v^i.
\end{align*}
We draw a sample $y \sim U[0,1]$ from the uniform distribution on $[0,1]$. If $y$ is in a section of the interval corresponding to $p_v^i$, with $i\not=0$, then we set the spin to $i$---if not, we recursively call the algorithm on the vertices in $S_\ell\setminus\Lambda$ in turn, each time conditioning on the assignments to the previous vertices in this set. (See Algorithm~\ref{alg:SSMS}.)

\begin{algorithm}[ht]
\DontPrintSemicolon
 \caption{$\textsc{ssms}(\mathcal S = (G = (V,E), q, b, A), (\Lambda, \sigma), v, \ell)$}\label{alg:SSMS}
 \KwIn{The algorithm takes a spin system $\mathcal S$, a set of known vertices $\Lambda \subset V$ with a configuration $\sigma \in \Omega_\Lambda$, a vertex to sample $v \notin \Lambda$, and a distance $\ell \in \mathbb N$.}
 \KwOut{The algorithm returns the partial configuration passed in with a spin at $v$ as well: $(\Lambda, \sigma) \oplus (v,i)$ for some $i \in [q]$.}
 $S_\ell \leftarrow \{w \in V : d_G(v,w) = \ell\}$ \;
\For{$i \in [q]$}{
   $p_v^i \leftarrow \min_{\tau \in \Omega_{S_\ell\setminus\Lambda}} \mu_v^{\sigma \oplus \tau}(i)$ \;
   $I_i \leftarrow \left[\sum_{j = 1}^{i-1} p_v^j, \sum_{j = 1}^i p_v^j\right)$ \;
   }
 $p_v^0 \leftarrow 1 - \sum_{i \in [q]} p_v^i$ \;
 $I_0 \leftarrow \left[1-p_v^0,1\right]$ \;
 Sample $y \sim U[0,1]$, a realisation of a uniform $[0,1]$ random variable. \;
 \eIf{$y \in I_0$}{
   $(J_1,...,J_q) \leftarrow \textsc{bd-split} (\mathcal S, (\Lambda, \sigma), v, \ell, (p_v^1,...,p_v^q))$ \;
   Find $i \in [q]$ such that $y \in J_i$ \;
   \textbf{return} $((\Lambda, \sigma) \oplus (v,i))$ \;
   }{
     Find $i$ such that $y \in I_i$ \;
     \textbf{return} $((\Lambda, \sigma) \oplus (v,i))$ \;
   }
\end{algorithm}

Once the spins on $S_\ell\setminus\Lambda$ (and hence $S_\ell\cup\Lambda$) are decided, we subdivide the interval~$I_0$ (the zone of indecision) into $q$~subintervals, whose lengths can now be calculated precisely in subroutine \textsc{bd-split}, by brute force.  (See Algorithm~\ref{alg:BD-SPLIT}.)  The spin at the original vertex~$v$ is now determined using the original sample $y$ from $U[0,1]$.  

\begin{algorithm}[ht]
\DontPrintSemicolon
 \caption{\textsc{bd-split} $(\mathcal S, (\Lambda, \sigma), v, \ell, (p_v^1,...,p_v^q))$}\label{alg:BD-SPLIT}
 Give $S_\ell(v)\setminus\Lambda$ an ordering $S_\ell(v)\setminus\Lambda = \{w_1,...,w_m\}$. \;
   $(\Lambda',\sigma') \leftarrow (\Lambda,\sigma)$\;
   \For{$j \in [1,...,m]$}{
     $(\Lambda',\sigma') \leftarrow $\textsc{ssms}$(\mathcal S, (\Lambda',\sigma'), w_j, \ell)$
     }
   \For{$i \in [q]$}{
     $\rho_v^i \leftarrow \mu_v^{\sigma'}(i) - p_v^i$ \;
     $I_i \leftarrow \left[\sum_{j = 1}^{q} p_v^j + \sum_{k = 1}^{i-1}\rho_v^i, \sum_{j = 1}^{q} p_v^j + \sum_{k = 1}^{i}\rho_v^i \right)$ \;
     }
   \textbf{return} ($I_1,...,I_q$)
\end{algorithm}

As we will see more concretely, there is a trade-off in this algorithm. As the radius~$\ell$ grows, the calculations to determine the subdivision of $I_0$ rapidly increase, but the expected number of recursive calls will decrease.

An important thing to notice about the algorithm is that it does not terminate for all choices of $\mathcal S$ and $\ell$, even if $G$ is finite.  We will see that the correctness of the algorithm essentially depends on whether it terminates with probability~1.

\subsection{Run-time}

Our run-time analysis is in terms of the number of calls of \textsc{ssms}. For fixed $\ell$, the amount of time each call takes will be $O(1)$.

Our viewpoint is that the algorithm behaves similarly to a branching process where each recursive call of the algorithm is a child of the vertex which initiated the call. This comparison is not exact of course, as the calls at different vertices will not branch according to one distribution, but it is enough to give us a broad class of situations where the algorithm terminates.

\begin{thm} \label{runtime}
	Let $\mathcal S = (G = (V,E),q,b,A)$ be a spin system which exhibits strong spatial mixing with rate $f$.  Suppose that we have a growth bound for~$G$ of the form 
	\begin{align*}
      |S_\ell(v)| \leq g(\ell),
    \end{align*}
%MRJ  moved the quantification to the right place.
    for all $v \in V$.
    Let $\ell \in \mathbb N$ be such that 
    \begin{align*}
      q \cdot f(\ell) \cdot g(\ell) \leq \alpha < 1,
    \end{align*}
    Then for all $v \in V$, $\Lambda \subset V$, and $\sigma \in \Omega_\Lambda$ such that
    \begin{align*}
      \mu_\Lambda(\sigma) > 0,
    \end{align*}
    the expected number of times \textsc{ssms}$(\mathcal S, (\Lambda,\sigma), v, \ell)$ calls \textsc{ssms} (in total) is bounded by $\frac1{1-\alpha}$.
\end{thm}

\begin{proof}
	Let us construct a random tree $T_A$ out of our recursive calls. Our root will be our initial call. Recursively, if at some node $w$ of the tree we fall in the zone of indecision, we add a child to the node for every recursive call that is initiated by \textsc{bd-split}. If we do not fall in the zone of indecision, then $w$ has no children.
	
	Let $T_B$ be the branching process with offspring distribution $\xi$ where
	\begin{align*}
      \mathbb P[\xi = 0] = 1 - q \cdot f(\ell), \quad \mathbb P[\xi = g(\ell)] = q \cdot f(\ell).
    \end{align*}
    Then $T_B$ stochastically dominates $T_A$. Indeed, at every node $w \in T_A$, the probability of falling in the zone of indecision is bounded by $q \cdot f(\ell)$ due to the spin system's strong spatial mixing. Further, when $w$ does have children, the number of children is bounded above by $|S_\ell(w)|\leq g(\ell)$.
    
    Now we can leverage the theory of branching processes to finish off our bound. We know that
    \begin{align*}
      \mathbb E [\xi] \leq \alpha < 1,
    \end{align*}
    so by the fundamental theorem of branching processes, we know that
    \begin{align*}
    	\mathbb E[|T_A|] \leq \mathbb E[|T_B|] \leq \frac1{1 - \alpha}. 
    \end{align*}
\end{proof}

From this, we see that under appropriate conditions we can sample sections of graphs of size $n$ in time $O(n)$. It is also worth noticing that the above conditions give termination of the algorithm with probability 1.  Note that, unless $\ell$ is very small, the computation required to compute the subdivision of $I_0$ in \textsc{bd-split} will be huge.  So although the expected time to compute one spin is constant, that constant may be impractically large.  This will be a particular issue when the system is close to non-uniqueness.  

The qualitative nature of the main result in this section is summarised in the following corollary.

\begin{cor} \label{general}
	Let $G = (V,E)$ be a graph with sub-exponential growth bounded by $g$ and let $\mathcal S = (G,q,b,A)$ be a spin system which exhibits strong spatial mixing with rate $f$ which is inversely exponential.
    Consider any choices of $v \in V$, $\Lambda \subset V$, and $\sigma \in \Omega_\Lambda$. Then %there exists an $\ell$ such that, defining $\tau$ to be 
    the running time $\tau$ of \textsc{ssms}$(\mathcal S, (\Lambda,\sigma), v, \ell)$ satisfies
    $\mathbb E[\tau] = O(1)$.
\end{cor}

\begin{proof}
In Theorem~\ref{runtime}, choose $\ell$ large enough to satisfy the inequality $q \cdot f(\ell) \cdot g(\ell) \leq \alpha < 1$.
\end{proof}

\subsection{Correctness}

Our understanding of the correctness of the algorithm depends on our view of the algorithm as a branching process.

\begin{thm} \label{correctness}
	Let $\mathcal S = (G = (V,E),q,b,A)$ be a spin system which exhibits strong spatial mixing with rate $f$. Suppose for all $v \in V$, all finite $\Lambda \subset V$, and all $\sigma \in \Omega_\Lambda$ such that
    \begin{align*}
      \mu_\Lambda(\sigma) > 0,
    \end{align*}
    \textsc{ssms}$(\mathcal S, (\Lambda,\sigma), v, \ell)$ terminates with probability 1. Then
    \begin{align*}
    	\mathbb P\left[\textsc{ssms}(\mathcal S, (\Lambda,\sigma), v, \ell) = i\right] = \mu_v^\sigma(i).
    \end{align*}
\end{thm}

%MRJ  Explanatory footnote added.
Now rather than directly relating the output distribution of the algorithm to the true probability, we want to pass between the two via a similar algorithm --- which is just a depth-bounded version of the original one --- that is more readily seen to be correct.\footnote{The issue with the original version of the algorithm is that the recursion it employs is not well founded, so we cannot directly prove correctness by induction, as we sould like to do.}

Assume we have an oracle $\mathcal O$ which takes the input of \textsc{ssms} and a spin $i$ and produces the true probability $\mu_v^\sigma(i)$.  (Since we have strong spatial mixing, this probability is well defined.)  For some height $h \in \mathbb N$, our new bounded algorithm will exactly match \textsc{ssms} until the recursion depth reaches~$h$. At that point, the algorithm will use the oracle on all nodes at the deepest level. This is Algorithm \ref{alg2}.

\begin{algorithm} \label{alg2}
\DontPrintSemicolon
 \caption{$\textsc{bounded-ssms}(\mathcal S, (\Lambda, \sigma), v, \ell, h)$}
 \KwIn{The algorithm takes a spin system $\mathcal S$, a set of known vertices $\Lambda \subset V$ with a configuration $\sigma \in \Omega_\Lambda$, a vertex to sample $v \notin \Lambda$, a distance $\ell \in \mathbb N,$ and a height $h \in \mathbb N$.}
 \KwOut{The algorithm returns the partial configuration passed in with a spin at $v$ as well: $(\Lambda, \sigma) \oplus (v,i)$ for some $i \in [q]$.}
 \eIf{$h = 0$}{
 Use $\mathcal O$ to sample $i \sim \mu_v^\sigma$ \\
 \textbf{return} $(\Lambda, \sigma) \oplus (v,i)$
 }
 {
 $S_\ell\leftarrow \{w \in V : d_G(v,w) = \ell\}$ \;
\For{$i \in [q]$}{
   $p_v^i \leftarrow \min_{\tau \in \Omega_{S_\ell\setminus\Lambda}} \mu_v^{\sigma \oplus \tau}(i)$ \;
   $I_i \leftarrow \left[\sum_{j = 1}^{i-1} p_v^j, \sum_{j = 1}^i p_v^j\right)$ \;
   }
 $p_v^0 \leftarrow 1 - \sum_{i \in [q]} p_v^i$ \;
 $I_0 \leftarrow \left[1-p_v^0,1\right]$ \;
 Sample $y \sim U[0,1]$, a realisation of a uniform $[0,1]$ random variable. \;
 \eIf{$y \in I_0$}{
   $(J_1,...,J_q) \leftarrow \textsc{bounded-bd-split} (\mathcal S, (\Lambda, \sigma), v, \ell, (p_v^1,...,p_v^q),S_\ell\setminus\Lambda,h-1)$ \;
   Find $i \in [q]$ such that $y \in J_i$ \;
   \textbf{return}$((\Lambda, \sigma) \oplus (v,i))$ \;
   }{
     Find $i$ such that $y \in I_i$ \;
     \textbf{return}$((\Lambda, \sigma) \oplus (v,i))$ \;
   }
   }
\end{algorithm}

\begin{algorithm}
\DontPrintSemicolon
 \caption{\textsc{bounded-bd-split} $(\mathcal S, (\Lambda, \sigma), v, \ell, (p_v^1,...,p_v^q),h)$}
 Give $S_\ell\setminus\Lambda$ an ordering $S_\ell\setminus\Lambda = \{w_1,...,w_m\}$. \;
   $(\Lambda',\sigma') \leftarrow (\Lambda,\sigma)$\;
   \For{$j \in [1,...,m]$}{
     $(\Lambda',\sigma') \leftarrow $\textsc{bounded-ssms}$(\mathcal S, (\Lambda',\sigma'), w_j, \ell, h)$
     }
   \For{$i \in [q]$}{
     $\rho_v^i \leftarrow \mu_v^{\sigma'}(i) - p_v^i$ \;
     $I_i \leftarrow \left[\sum_{j = 1}^{q} p_v^j + \sum_{k = 1}^{i-1}\rho_v^i, \sum_{j = 1}^{q} p_v^j + \sum_{k = 1}^{i}\rho_v^i \right)$ \;
     }
   \textbf{return} ($I_1,...,I_q$)
\end{algorithm}

Now we see the correctness:

\begin{lem}
	Under the conditions of Theorem~\ref{correctness}, for all $h \in \mathbb N$ and for all $i \in [q]$
	\begin{align*}
      \mathbb P\left[\textsc{bounded-ssms}(\mathcal S, (\Lambda, \sigma),v,\ell,h) = i \right] = \mu_v^\sigma(i).
    \end{align*}
\end{lem}

\begin{proof}
	We prove this by induction on $h$. For the base case of $h = 0$ the lemma is true due to the use of the oracle $\mathcal O$. Now suppose it holds holds for $h = k$.  Given the order $S_\ell(v)\setminus\Lambda = \{w_1,...,w_m\}$ that the algorithm processes the vertices in, we have
%	\begin{align*}
%      &\ \mathbb P\left[\textsc{bounded-ssms}(\mathcal S, (\Lambda, \sigma),v,\ell,k+1) = i \right] \\
%      = &\ p_v^i + \\
%      &\quad\sum_{\tau \in \Omega_{S_\ell\setminus\Lambda}} \left(\mu_v^{\sigma \oplus \tau}(i) - p_v^i\right) \prod_{j=1}^m \mathbb P\left[\textsc{bounded-ssms}\left(\mathcal S, (\Lambda, \sigma) \oplus \bigoplus_{k = 1}^{j-1}(w_k,\tau_{w_k}),w_j,\ell,k\right) = \tau_{w_j} \right] \\
%      = &\ p_v^i + \sum_{\tau \in \Omega_{S_\ell\setminus\Lambda}} \left(\mu_v^{\sigma \oplus \tau}(i) - p_v^i\right) \prod_{j=1}^m \mu_{w_j}^{\left(\Lambda, \sigma) \oplus \bigoplus_{k = 1}^{j-1}(w_k,\tau_{w_k}\right)}(\tau_{w_j}),
%    \end{align*}
%MRJ  Some (v)'s added.
	\begin{align*}
      &\mathbb P\left[\textsc{bounded-ssms}(\mathcal S, (\Lambda, \sigma),v,\ell,k+1) = i \right] \\
      &\quad = p_v^i + 
      \sum_{\tau \in \Omega_{S_\ell(v)\setminus\Lambda}} \left(\mu_v^{\sigma \oplus \tau}(i) - p_v^i\right)\\
      &\qquad\qquad\qquad\times\prod_{j=1}^m \mathbb P\left[\textsc{bounded-ssms}\left(\mathcal S, (\Lambda, \sigma) \oplus \bigoplus_{k = 1}^{j-1}(w_k,\tau_{w_k}),w_j,\ell,k\right) = \tau_{w_j} \right] \\
      &\quad = p_v^i + \sum_{\tau \in \Omega_{S_\ell(v)\setminus\Lambda}} \left(\mu_v^{\sigma \oplus \tau}(i) - p_v^i\right) \prod_{j=1}^m \mu_{w_j}^{\left(\Lambda, \sigma) \oplus \bigoplus_{k = 1}^{j-1}(w_k,\tau_{w_k}\right)}(\tau_{w_j}),
    \end{align*}
    where the second inequality comes from the induction hypothesis. Finally, observing that
    \begin{align*}
      p_v^i = \sum_{\tau \in \Omega_{S_\ell(v)\setminus\Lambda(v)}} p_v^i \cdot \prod_{j=1}^m \mu_{w_j}^{\left(\Lambda, \sigma) \oplus \bigoplus_{k = 1}^{j-1}(w_k,\tau_{w_k}\right)}(\tau_{w_j}),
    \end{align*}
    we see that
    \begin{align*}
      &\mathbb P\left[\textsc{bounded-ssms}(\mathcal S, (\Lambda, \sigma),v,\ell,k+1) = i \right] \\
      &\qquad=  p_v^i + \sum_{\tau \in \Omega_{S_\ell(v)\setminus\Lambda(v)}} \left(\mu_v^{\sigma \oplus \tau}(i) - p_v^i\right) \prod_{j=1}^m \mu_{w_j}^{\left(\Lambda, \sigma) \oplus \bigoplus_{k = 1}^{j-1}(w_k,\tau_{w_k}\right)}(\tau_{w_j}) \\
      &\qquad= \sum_{\tau \in \Omega_{S_\ell(v)\setminus\Lambda}} \mu_v^{\sigma \oplus \tau}(i) \prod_{j=1}^m \mu_{w_j}^{\left(\Lambda, \sigma) \oplus \bigoplus_{k = 1}^{j-1}(w_k,\tau_{w_k}\right)}(\tau_{w_j}) \\
      &\qquad= \mu_v^\sigma (i).
    \end{align*}
\end{proof}

Now we use this lemma to bridge the gap between the algorithm and the true probability.

\begin{proof}[Proof of Theorem \ref{correctness}]
	Let $\nu_h$ and $\nu$ be the distributions of \textsc{bounded-ssms}$(\mathcal S, (\Lambda, \sigma),v,\ell,h)$ and \textsc{ssms}$(\mathcal S, (\Lambda, \sigma),v,\ell)$ respectively.
	Recall that
	\begin{align*}
	  d_{\textnormal{TV}}(\nu_h, \nu) = \inf_{\substack{X \sim \nu, \\ Y \sim \nu_h}} \mathbb P[X \not= Y],
    \end{align*}
    where the random variables $X$ and $Y$ are on a joint sample space.  
    %By using the same random variables in both algorithms, we immediately see that
    While the depth of recursion remains below $h$, we can couple the random choices made by the bounded and unbounded versions of the algorithm. Thus, using $T_A$ and $T_B$ from the run-time argument,
    \begin{align*}
      d_{\textnormal{TV}}(\nu, \nu_h) &\leq \mathbb P[\textnormal{height}(T_A) \geq h] \\
      &\leq \mathbb P[\textnormal{height}(T_B) \geq h] \\
      &\ensuremath{\stackrel{h \to \infty}{\longrightarrow}} 0.
    \end{align*}
    Now since $\nu_h$ is also the correct distribution, we have that for all $h$
    \begin{align*}
      &|\mathbb P\left[\textsc{ssms}(\mathcal S, (\Lambda,\sigma), v, \ell) = i\right] - \mu_v^\sigma(i)| \leq d_\textnormal{TV}(\nu,\nu_h) 
    \end{align*}
    which implies
    $\mathbb P\left[\textsc{ssms}(\mathcal S, (\Lambda,\sigma), v, \ell) = i\right] = \mu_v^\sigma(i)$, as required.
\end{proof}

\section{Examples}
\subsection{The Hardcore model}
In Section~\ref{sec:hardcore} we informally analysed the radius 1 version of the algorithm.  We return to this case as an example application of Theorem~\ref{correctness}.  Setting $\ell=1$, we have $f(1)=\lambda/(1+\lambda)$, $g(1)=\Delta$ and $q=2$.  Then the algorithm has constant expected run-time provided $q\,f(\ell)g(\ell)<1$, which holds when $\lambda<1/(2\Delta-1)$. This is worse than the condition that arose in our informal analysis, owing to some slack in Lemma~\ref{lem:pvzerobound}.  

Now we use Theorem~\ref{general} to find significantly more flexibility. On general graphs of maximum degree $\Delta$ we know from Weitz \cite{weitz2006counting} that the hardcore model exhibits strong spatial mixing whenever 
\begin{align*}
	\lambda < \lambda_c: = \frac{(\Delta-1)^{\Delta-1}}{(\Delta-2)^\Delta}.
\end{align*}
Let $\mathcal{G}$ be a class of graphs of maximum degree $\Delta$ and subexponental growth.
It follows from Theorem~\ref{general} that for all $\lambda < \lambda_c$ and large enough $\ell$ (depending on $\lambda$) the expected run-time of the algorithm is constant per assigned spin.  It is worth noting that for the case of a general graph of maximum degree $\Delta$ it was recently shown, by Chen, Liu and Vigoda~\cite{Chen2020OptimalMO}, that Glauber dynamics mixes in time $O(n \log n)$ when $\lambda < \lambda_c$.  Thus, for general graphs our algorithm is efficient in the same regime as the Glauber dynamics, but we do require the additional condition of subexponential growth.  On the plus side, our algorithm is linear in the number of assigned spins, works for infinite graphs, and produces a perfect sample.  Spinka~\cite[Cor.~1.6]{spinka2020finitary} also shows how to produce perfect samples when $\lambda < \lambda_c$, though his proofs are restricted to the square lattices $\mathbb{Z}^d$.   (It could well be the case that Spinka's approach works in a more general setting.) 

For a more concrete example, consider the integer lattice $\mathbb Z^2$. In the radius 1 case we have efficiency for $\lambda < \frac13$, but by leveraging strong spatial mixing our range increase to at least $\lambda < 2.538$, a bound that was derived by Sinclair, Srivastava, \v Stefankovi\v c and  Yin \cite{sinclair2017spatial}. In fact, the true threshold of strong spatial mixing on $\mathbb Z^2$ is believed to be even higher: approximately 3.7962 \cite{baxter1980hard, vera2013improvedjournal}.  Note that the graphs in this example have subexponental (indeed, polynomial) growth.  In the case of $\mathbb Z^2$, Spinka's approach~\cite{spinka2020finitary} has the same range of validity as ours.

\subsection{The Monomer-Dimer model}

Another widely studied model which can be viewed as a spin system is the \textit{monomer-dimer} model. When we restrict ourselves to a finite graph $G = (V,E)$, the monomer-dimer model with dimer activity $\gamma$ is the distribution over matchings on the graph with the probability of a matching~$M$ being proportional to $\gamma^{|M|}$. To view this model as a spin system we simply consider the line-graph of $G$.

\begin{dfn}
	Let $G = (V,E)$ be a finite graph. The line-graph $L(G)$ is the graph with one vertex corresponding to each edge in $G$ and an edge between any $e, e' \in E$ which are incident.
\end{dfn}

The monomer-model can then be represented as the hardcore model on $L(G)$ with $\lambda = \gamma$. The monomer-dimer model exhibits strong spatial mixing when $G$ has sub-exponential neighbourhood growth --- a fact essentially established by van den Berg and Brouwer~\cite{vandenBergBrouwer} and extended by Bayati et al.~\cite{bayati2007} --- so returning to Theorem~\ref{general} we see that \textsc{ssms} allows us to sample partial configurations of the monomer-dimer model at any fixed $\lambda$ in linear time, at least for graphs of subexponential growth.

Under the same conditions, there already exists a sampling algorithm, based on Glauber dynamics, that runs in polynomial time \cite{Jerrum1989ApproximatingTP}.  Indeed, recent work of Chen et al.~\cite{Chen2020OptimalMO} shows that Glauber dynamics mixes in time $O(n\log n)$.  The advantage of the current algorithm is that it is exact and we can sample just a partial configuration on the graph.

\subsection{The $q$-colour Model}

Another well-known spin system model is the $q$-colour model where each spin corresponds to a colour. Here we define for every $i,j \in [q]$
\begin{align*}
  b(i) = 1, \quad A(i,j) = \I{i \not= j},
\end{align*}
i.e., a colour cannot be adjacent to itself but every other configuration is equally weighted. By Theorem~\ref{general}, whenever we have a graph~$G$ (more precisely, a class of graphs) with sub-exponential neighbourhood growth on which the $q$-colour model exhibits strong spatial mixing, \textsc{ssms} will allow us to sample partial configurations of the $q$-colour model on the graph in linear time.

One such regime is revealed by Goldberg, Martin and Paterson \cite{Goldberg2005StrongSM}. In addition to our neighbourhood growth condition, if $G$ is triangle-free with maximum degree $\Delta \geq 3$ then we have strong spatial mixing when
\begin{align*}
  q > \alpha \Delta - \gamma \approx 1.76\Delta - 0.47,
\end{align*}
where $\alpha^\alpha = e$ and $\gamma = \frac{4\alpha^3 - 6\alpha^2 - 3\alpha + 4}{2(\alpha^2 - 1)}$.

Another strong spatial mixing regime is described by Efthemiou et al.~\cite{efthymiou2019improved}. Here, for the $(d+1)$-regular tree, strong spatial mixing holds when $q > 1.59 d$.  This result raises hope that the dependence of $q$ on $\Delta$ can be improved in due course.

\subsection{The ferromagnetic Ising model on the 2-dimensional square lattice}

The ferromagnetic Ising model (in the absence of an external field) is a spin model characterised by the following interactions:
\begin{align*}
  b = 
\begin{pmatrix}
  1 \\
  1
\end{pmatrix}, \quad A =
\begin{pmatrix}
  \lambda && 1 \\
  1 && \lambda
\end{pmatrix},
\end{align*}
with $\lambda\geq1$.  It is known that this model undergoes a phase transition at $\lambda_c=1+\sqrt2$.  When $\lambda<\lambda_c$ there is a unique Gibbs measure, while for $\lambda>\lambda_c$ there are two.  On the 2-dimensional square lattice the model exhibits strong spatial mixing in the uniqueness region;  this is a consequence of a result of Martinelli, Olivieri and Schonmann~\cite{Martinelli1994weakstrong} that bootstraps weak spatial mixing to strong spatial mixing in 2-dimensional systems.  As a consequence, by Theorem~\ref{runtime}, we can perfectly sample partial configurations from the ferromagnetic Ising model on the 2d square lattice when $\lambda<\lambda_c$.   Note that this task was first achieved by van den Berg and Steif~\cite{vdBerg1999codings} using a more complicated algorithm inspired by Propp and Wilson's Coupling From The Past (CFTP).  Spinka~\cite{spinka2020finitary} also covers this situation in the context of a more general treatment.
%Mossel and Sly~\cite{Mossel2013exact} say that it follows (even for $d$-dimensional square lattices) from their Theorem~3.  
%For $d=2$ it seems to follow from  Blamca, Caputo, Sinclair and Vigoda state without reference that SSM holds for the Potts model on $\mathbb{Z}^2$, possibly using the Martinelli et al.\ machinery. 

One may wonder whether it is possible to sample perfectly from one of the two measures that exist when $\lambda>\lambda_c$.  A priori, there is no obvious obstacle.  However, van den Berg and Steif~\cite{vdBerg1999codings} show that this task is impossible on information-theoretic grounds, even with unbounded computational resources.

\section{Future Work}

While strong spatial mixing is sufficient for efficiency of the algorithm, we have not proven that it is necessary. Indeed, in the analysis there remains the possibility that strong spatial mixing is overkill for ensuring efficiency---strong spatial mixing does not tell us about the average proportion of bad events in the algorithm, but just gives us the worst case proportion of bad events.

The question then is in what ways we can relax this condition. A few possible directions include

\begin{itemize}
	\item Is there a condition which exactly captures this average case of bad events?
	\item What do conditions on properties like the connective constant tell us about the behaviour of the algorithm?
	\item What properties of a spin system do we require in general for linear time algorithms?
\end{itemize}

\printbibliography
\end{document}